\documentclass[letterpaper,twoside]{article}

\pdfoutput=1

\usepackage[english]{babel}
\usepackage{amsmath, amsthm, amssymb}
\usepackage[boxed,lined]{algorithm2e}
\usepackage{subfigure}
\usepackage{graphicx}
\usepackage[breaklinks=true,plainpages=false]{hyperref}
\usepackage{cite}
\usepackage{fullpage}

\newtheorem{problem}{Problem}
\newtheorem{lemma}{Lemma}
\newtheorem{theorem}{Theorem}
\newtheorem{corollary}{Corollary}
\newtheorem{defn}{Definition}
\newtheorem{fact}{Fact}

\newcommand{\M}{\texttt{M}}
\newcommand{\ICM}[1]{\texttt{ICM}$_{#1}$}
\newcommand{\LTM}[1]{\texttt{LTM}$_{#1}$}

\newcommand{\maxleaf}{\textrm{M{\scriptsize AX}\-L{\scriptsize EAF}}}
\newcommand{\strategyml}{\textrm{S{\scriptsize TRATEGY}\-M{\scriptsize AX}\-L{\scriptsize EAF}}}

\newcommand{\LS}{\textrm{L{\scriptsize OCAL}\-S{\scriptsize EARCH}}}
\newcommand{\rand}{\textrm{R{\scriptsize ANDOM}\-P{\scriptsize RICING}}}

\newcommand{\footnoteremember}[2]{
  \footnote{#2}
  \newcounter{#1}
  \setcounter{#1}{\value{footnote}}
}
\newcommand{\footnoterecall}[1]{
  \footnotemark[\value{#1}]
}

\begin{document}

\title{Pricing strategies for viral marketing on Social Networks}

\author{
David Arthur\footnoteremember{cs}{Department of Computer Science, Stanford
  University} \qquad
Rajeev Motwani\footnoterecall{cs} \qquad
Aneesh Sharma\footnote{Institute for Computational and Mathematical
  Engineering, Stanford University} \qquad
Ying Xu\footnote{Work done when author was a student at the Department of
  Computer Science, Stanford University} \\
{\tt \{darthur,rajeev,aneeshs,xuying\}@cs.stanford.edu}
}

\maketitle

\begin{abstract}
  We study the use of viral marketing strategies on social networks to
  maximize revenue from the sale of a single product. We propose a model in
  which the decision of a buyer to buy the product is influenced by friends
  that own the product and the price at which the product is offered. The
  influence model we analyze is quite general, naturally extending both the
  Linear Threshold model and the Independent Cascade model, while also
  incorporating price information.  We consider sales proceeding in a
  cascading manner through the network, i.e. a buyer is offered the product
  via recommendations from its neighbors who own the product. In this
  setting, the seller influences events by offering a cashback to
  recommenders and by setting prices (via coupons or discounts) for each
  buyer in the social network.

  Finding a seller strategy which maximizes the expected revenue in this
  setting turns out to be NP-hard. However, we propose a seller strategy that
  generates revenue guaranteed to be within a constant factor of the optimal
  strategy in a wide variety of models. The strategy is based on an {\em
    influence-and-exploit} idea, and it consists of finding the right
  trade-off at each time step between: generating revenue from the current
  user versus offering the product for free and using the influence generated
  from this sale later in the process.  We also show how local search can be
  used to improve the performance of this technique in practice.
\end{abstract}

\section{Introduction}
Social networks such as Facebook, Orkut and MySpace are free to join, and
they attract vast numbers of users. Maintaining these websites
for such a large group of users requires substantial investment from the host
companies.  To help recoup these investments, these companies often turn to
monetizing the information that their users provide for free on these
websites. This information includes both detailed profiles of users and also
the network of social connections between the users.  Not surprisingly, there
is a widespread belief that this information could be a gold mine for
targeted advertising and other online businesses.  Nonetheless, much of this
potential still remains untapped today. Facebook, for example, was valued at
\$15 billion by Microsoft in 2007~\cite{BBC07}, but its estimated revenue in
2008 was only \$300 million~\cite{Wikipedia:Facebook}. With so many users and so
much data, higher profits seem like they should be possible.  Facebook's
Beacon advertising system does attempt to provide targeted advertisements but
it has only obtained limited success due to privacy
concerns~\cite{Wikipedia:Beacon}.

This raises the question of how companies can better monetize the already public
data on social networks without requiring extra information and thereby compromising
privacy. In particular, most large-scale monetization technologies
currently used on social networks are modeled on the sponsored search paradigm of
contextual advertising and do not effectively leverage the networked nature of the data.

Recently, however, people have begun to consider a different monetization
approach that is based on selling products through the spread of
influence. Often, users can be convinced to purchase a product if many of
their friends are already using it, even if these same users would be hard to
convince through direct advertising. This is often a result of personal
recommendations -- a friend's opinion can carry far more weight than an
impersonal advertisement. In some cases, however, adoption among friends is
important for even more practical reasons. For example, instant messenger
users and cell phone users will want a product that allows them to talk
easily and cheaply with their friends.  Usually, this encourages them to
adopt the same instant messenger program and the same cell phone carrier that
their friends have.  We refer the reader to previous work and the references
therein for further explanations behind the motivation of the influence
model~\cite{Kleinberg_AGT07,HMS08}.

In fact, many sellers already do try to utilize influence-and-exploit strategies that are based on
these tendencies. In the advertising world, this has recently led to the adoption
of {\em viral marketing}, where a seller attempts to artificially create word-of-mouth
advertising among potential customers~\cite{LSK06,LAH07,Schon08}. A more powerful but riskier
technique has been in use much longer: the seller gives out free samples
or coupons to a limited set of people, hoping to convince these people to try out the product and then
recommend it to their friends. Without any extra data, however, this forces sellers
to make some very difficult decisions. Who do they give the free samples to? How many free
samples do they need to give out? What incentives can they afford to give to recommenders
without jeopardizing the overall profit too much?

In this paper, we are interested in finding systematic answers to these questions. In general
terms, we can model the spread of a product as a process on a social network. Each node
represents a single person, and each edge represents a friendship. Initially, one or more
nodes is ``active'', meaning that person already has the product. This could either be a large set of
nodes representing an established customer base, or it could be just one node -- the seller --
whose neighbors consist of people who independently trust the seller, or who are otherwise
likely to be interested in early adoption.

At this point, the seller can encourage the spread of influences in two ways. First of all, it can offer
cashback rewards to individuals who recommend the product to their friends. This is often seen in practice
with ``referral bonuses'' -- each buyer can optionally name the person who referred them, and this person
then receives a cash reward. This gives existing buyers an incentive to recommend the product to
their friends. Secondly, a seller can offer discounts to specific people in order to encourage them to
buy the product, above and beyond any recommendations they receive. It is important to choose a good discount from
the beginning here. If the price is not acceptable when a prospective buyer first receives recommendations,
they might not bother to reconsider even if the price is lowered later.

After receiving discount offers and some set of recommendations, it is up to
the prospective buyers to decide whether to actually go through with a
purchase. In general, they will do so with some probability that is
influenced by the discount and by the set of recommendations they have
received. The form of this probability is a parameter of the model and it is
determined by external factors, for instance, the quality of the product and
various exogenous market conditions. While it is impossible for a seller to
calculate the form of these probability exactly, they can estimate it from
empirical observations, and use that estimate to inform their policies. One
could interpret the probabilities according to a number of different models
that have been proposed in the literature (for instance, the Independent
Cascade and Linear Threshold models), and hence it is desirable for the
seller to be able to come up with a strategy that is applicable to a wide
variety of models.

Now let us suppose that a seller has access to data from a social network such as
Facebook, Orkut, or MySpace. Using this, the seller can estimate what the real, true,
underlying friendship structure is, and while this estimate will not be perfect, it is
getting better over time, and any information is better than none. With this information
in hand, a seller can model the spread of influence quite accurately, and the formerly
inscrutable problems of who to offer discounts to, and at what price, become algorithmic
questions that one can legitimately hope to solve. For example, if a seller knows the
structure of the network, she can locate individuals that are particularly well connected
and do everything possible to ensure they adopt the product and exert their considerable influence.

In this paper, we are interested in the algorithmic side of this question: Given the network
structure and a model of the purchase probabilities, how should
the seller decide to offer discounts and cashback rewards?

\subsection{Our contributions}
We investigate seller strategies that address the above questions
in the context of expected revenue maximization. We will focus much of our
attention on {\em non-adaptive} strategies for the seller: the seller chooses and
commits to a discount coupon and cashback offer for each potential buyer before the
cascade starts. If a recommendation is given to this node at any time,
the price offered will be the one that the seller committed to initially, irrespective
of the current state of the cascade.

A wider class of strategies that one could consider are
{\em adaptive strategies}, which do not have this restriction. For example, in an
adaptive strategy, the seller could choose to observe the outcome of the (random)
cascading process up until the last minute before making very well informed
pricing decisions for each node. One might imagine that this additional flexibility
could allow for potentially large improvements over non-adaptive strategies.
Unfortunately, there is a price to be paid, in that good adaptive strategies are likely
to be very complicated, and thus difficult and expensive to implement. The ratio of
the revenue generated from the optimal adaptive strategy to the revenue generated
from the optimal non-adaptive strategy is termed the ``adaptivity gap''.

Our main theoretical contribution is a very efficient non-adaptive strategy whose expected revenue is
within a constant factor of the optimal revenue from an {\em adaptive} strategy. This
guarantee holds for a wide variety of probability functions, including natural extensions
of both the Linear Threshold and Independent Cascade models\footnote{More precisely, the
strategy achieves a constant-factor approximation for any {\em fixed} model, independent
of the social network. If one changes the model, the approximation factor does vary, as
made precise in Section~\ref{sec:approxalg}.}. Note that a surprising consequence of this result
is that the adaptivity gap is constant, so one can make the case that not much is lost by
restricting our attention to non-adaptive policies. We also show that the problem of
finding an optimal non-adaptive strategy is NP-hard, which means an efficient
approximation algorithm is the best theoretical result that one could hope for.

Intuitively, the seller strategy we propose is based on an {\em influence-and-exploit} idea,
and it consists of categorizing each potential buyer as either an {\em influencer} or a
{\em revenue source}. The influencers are offered the product for free and the revenue
sources are offered the product at a pre-determined price, chosen based on the exact
probability model. Briefly, the categorization is done by finding a spanning tree of
the social network with as many leaves as possible, and then marking the leaves as
revenue sources and the internal nodes as influencers. We can find such a tree
in near-linear time~\cite{KW91,LR98}. Cashback amounts are chosen to be a fixed fraction of
the total revenue expected from this process. The full details are presented in
section~\ref{sec:approxalg}.

In practice, we propose using this approach to find a strategy that has good
global properties, and then using local search to improve it further. This kind of combination has
been effective in the past, for example on the k-means problem \cite{AV07}. Indeed, experiments
(see section~\ref{sec:localsearch}) show that combining local search with the above influence-and-exploit
strategy is more effective than using either approach on its own.

\subsection{Related work}
The problem of {\em social contagion} or spread of influence was first
formulated by the sociological community, and introduced to the computer
science community by Domingos and Richardson~\cite{DR01}. An influential
paper by Kempe, Kleinberg and Tardos~\cite{KKT03} solved the {\em target set
  selection} problem posed by~\cite{DR01} and sparked interest in this area
from a theoretical perspective (see~\cite{Kleinberg_AGT07}). This work has
mostly been limited to the {\em influence maximization} paradigm, where
influence has been taken to be a proxy for the revenue generated through a
sale. Although similar to our work in spirit, there is no notion of price in this
model, and therefore, our central problem of setting prices to encourage influence
spread requires a more complicated model.

A recent work by Hartline, Mirrokni and Sundararajan~\cite{HMS08} is similar
in flavor to our work, and also considers extending social contagion ideas
with pricing information, but the model they examine differs from our model
in a several aspects. The main difference is that they assume that the seller
is allowed to approach arbitrary nodes in the network at any time and offer
their product at a price chosen by the seller, while in our model the cascade
of recommendations determines the timing of an offer and this cannot be
directly manipulated. In essence, the model proposed in~\cite{HMS08} is akin
to advertising the product to arbitrary nodes, bypassing the network
structure to encourage a desired set of early adopters. Our model restricts
such direct advertising as it is likely to be much less effective than a
direct recommendation from a friend, especially when the recommender has an
incentive to convince the potential buyer to purchase the product (for
instance, the recommender might personalize the recommendation, increasing
its effectiveness). Despite the different models, the algorithms proposed by
us and~\cite{HMS08} are similar in spirit and are based on an {\em
  influence-and-exploit} strategy.

This work has also been inspired by a direction mentioned by
Kleinberg~\cite{Kleinberg_AGT07}, and is our interpretation of the informal
problem posed there. Finally, we point out that the idea of cashbacks has
been implemented in practice, and new retailers are also embracing the
idea~\cite{LSK06,LAH07,Schon08}. We note that some of the systems being
implemented by retailers are quite close to the model that we propose, and
hence this problem is relevant in practice.

\section{The Formal Model}
\label{sec:model}
Let us start by formalizing the setting stated above. We represent
the social network as an undirected graph $G(V,E)$, and denote the
initial set of adopters by $S^0\subseteq V$. We also denote the
{\em active} set at time $t$ by $S^{t-1}$ (we call a node {\em
active} if it has purchased the product and {\em inactive}
otherwise). Given this setting, the recommendations cascade through
the network as follows: at each time step $t\geq 1$, the nodes that
became active at time $t-1$ (i.e. $S^0$ for $t=1$, and
$u\in S^{t-1}\setminus S^{t-2}$ for $t\geq 2$) send recommendations
to their currently inactive friends in the network: $N^{t-1}=\{v\in V\setminus
S^{t-1}|(u,v)\in E, u\in S^{t-1}\setminus
S^{t-2}\}$. Each such node $v \in N^{t-1}$
is also given a price $c_{v,t}\in\mathbb{R}$ at which it can purchase the
product. This price is chosen by the seller to either be full price
or some discounted fraction thereof.

The node $v$ must then decide whether to purchase the product or not
(we discuss this aspect in the next section). If $v$ does accept the
offer, a fixed cashback $r>0$ is given to a recommender
$u\in S^{t-1}$ (note that we are fixing the cashback to be a positive constant for all
the nodes as the nodes are assumed to be non-strategic and any positive
cashback provides incentive for them to provide recommendations). If there are multiple recommenders, the buyer must
choose only one of them to receive the cashback; this is a system that is
quite standard in practice. In this way, offers are made to all nodes $v\in N^{t-1}$ through
the recommendations at time $t$ and these nodes make a decision at
the end of this time period. The set of active nodes is then updated
and the same process is repeated until the process quiesces, which
it must do in finite time since any step with no purchases ends the
process.

In the model described above, the only degree of freedom that the
seller has is in choosing the prices and the cashback amounts. It wants
to do this in a way that maximizes its own expected revenue
(the expectation is over randomness in the buyer strategies). Since
the seller may not have any control over the seed set, we are
looking for a strategy that can maximize the expected revenue starting
from any seed set on any graph. In most
online scenarios, producing extra copies of the product has negligible
cost, so maximizing expected revenue will also maximize expected profit.

Now we can formally state the problem of finding a revenue
maximizing strategy as follows:
\begin{problem} \label{problem:strategy}
  Given a connected undirected graph $G(V,E)$, a seed set $S^0$, a
  fixed cashback amount $r$, and a model \M{} for determining when
  nodes will purchase a product, find a strategy that
  maximizes the expected revenue from the cascading process
  described above.
\end{problem}
We are particularly interested in non-adaptive policies, which
correspond to choosing a price for each node in advance, making the price
independent of the time of the recommendation and the state of the cascade at
the time of the offer. Our goal
will be threefold: (1) to show that this problem is NP-hard even for
simple models \M, (2) to construct a constant-factor approximation
algorithm for a wide variety of models, and (3) to show that
restricting to non-adaptive policies results in at most a constant
factor loss of profit.

To simplify the exposition, we will assume the cashback $r = 0$ for
now. At the end of Section 4, we will show how the results can be
generalized to work for positive $r$, which should be sufficient
incentive for buyers to pass on recommendations.

\subsection{Buyer decisions}
\label{sec:BuyerDecision}
In this section, we discuss how to model the probability that a node
will actually buy the product given a set of recommendations and a
price. We use a very general model in this work that naturally extends the most
popular traditional models proposed in the influence maximization literature, including both Independent Cascade and Linear
Threshold.

Consider an abstract model \M{} for determining the probability
that a node will buy a product given a price and what
recommendations it has received. We allow \M{} to take on virtually any
form, imposing only the following conditions:
\begin{enumerate}
    \item The seller has full information about \M. This is a standard assumption, and
    it can be approximated in practice by running experiments and
    observing people's behavior.
    \item A node will never pay more than full price for the product
    (we assume this full price is 1 without loss of generality). Without an
    assumption like this, the seller could potentially achieve
    unbounded revenue on a single network, which makes the
    problem degenerate.
    \item A node will always accept the product and recommend it to
    friends if it receives a recommendation with price 0 (i.e. if a friend
    offers the product for free). Since nodes are given positive cash rewards
    for making recommendations, this condition is true for any rational buyer.
    \item If the social network is a single line graph with $S^0$ being the
    two endpoints, the maximum expected revenue is at most a constant $L$.
    Intuitively, this states that each prospective buyer on a social network should have
    some chance of rejecting the product (unless it's given to them for free),
    and therefore the maximum revenue on a line is bounded by a geometric series, and
    is therefore constant.
    \item There exist constants $f$, $c$, $q$ so that if more than
    fraction $f$ of a given node's neighbors recommend the product to the
    node at cost $c$, the node will purchase the product with
    probability $q$. This rules out extreme inertia, for example the case where no buyer will consider
    purchasing a product unless almost all of its neighbors have already done so.
\end{enumerate}
The fourth and fifth conditions here are used to parametrize how complicated
the model is, and our final approximation bound will be in terms of this model
``complexity'', which is defined to be $\frac{L}{(1-f)cq}$. While it may not
be obvious that all these conditions are met in general, we will show that they
are for both the Independent Cascade and Linear Threshold models, and indeed,
the arguments there extend naturally to many other cases as well.

In the traditional Independent Cascade model, there is a fixed
probability $p$ that a node will purchase a product each time it is
recommended to them. These decisions are made independently for each
recommendation, but each node will buy the product at most once.

To generalize this to multiple prices, it is natural to make $p$ a
function $[0,1]\rightarrow[0,1]$ where $p(x)$ represents the
probability that a node will buy the product at price $x$. For
technical reasons, however, it is convenient to work with the
inverse of $p$, which we call $C$.\footnote{It is sometimes useful
to consider functions $p(\cdot)$ that are not one-to-one. These functions have
no formal inverse, but in this case, $c$ can still be formally defined as
$C(x) = \max_y | p(y) | \ge x$.} Our general conditions on the
model reduce to setting $C(0) = 1$ and $C(1) = 0$ in this case. To ensure
bounded complexity, we also impose a minor smoothness condition.
\begin{defn}
    Fix a cost function $C: [0,1]\rightarrow[0,1]$ with
    $C(0) = 1, C(1) = 0$ and with $C$ differentiable at 0 and 1.
    We define the {\em Independent Cascade Model} \ICM{c} as follows:\\
    Every time a node receives a recommendation at price $C(x)$, it
    buys the product with probability $x$ and does nothing
    otherwise. If a node receives multiple recommendations, it
    performs this check independently for each recommendation but
    it never purchases the product more than once.
\end{defn}

\begin{lemma} \label{ICMLem}
    Fix a cost function $C$. Then:
    \begin{enumerate}
      \item \ICM{C} has bounded (model) complexity.
      \item If $C$ has maximum slope $m$ (i.e.
        $|C(x) - C(y)| \leq m|x-y|$ for all $x,y$), then $ICM_C$ has $O(m^2)$
        complexity.
      \item If $C$ is a step function with $n$ regularly spaced
            steps (i.e. $C(x) = C(y)$ if $\lfloor \frac{x}{n}
            \rfloor = \lfloor \frac{y}{n} \rfloor$), then \ICM{C}
            has $O(n^2)$ complexity.
    \end{enumerate}
\end{lemma}

\begin{proof}
    We show that the complexity of \ICM{C} can be bounded in terms
    of the maximum slope of $C$ near 0 and 1. Recall that if $C$ is
    differentiable at 0, then, by definition, there exists $\epsilon
    > 0$ so that
    $\frac{|C(x) - C(0)|}{x} \le |C'(0)| + 1$ for $x < \epsilon$. A similar
    argument can be made for $x = 1$, and thus we can say formally that
    there exist $m$ and $\epsilon>0$ such that:
    \begin{eqnarray*}
        &C(x) \ge 1 - mx& \textrm{ for } x \le \epsilon, \textrm{and}\\
        &C(x) \le m(1 - x)& \textrm{ for } x \ge 1-\epsilon.
    \end{eqnarray*}
    In this case, we will show that \ICM{C} has complexity at most
    $8\max(\frac{1}{\epsilon}, m)^2$, proving part 1. Note that parts 2 and 3 of the lemma will also follow
    immediately.

    We begin by analyzing $L_n$, the maximum expected revenue that
    can be achieved on a path of length $n$ if {\em one} of the
    endpoints is a seed. Note that $L \le 2\max_n L_n$ since selling a
    product on a line graph with two seeds can be thought of as
    two independent sales, each with one seed, that are cut short if
    the sales ever meet. Now we have:
    \begin{eqnarray*}
        L_n = \max_x x(C(x) + L_{n-1}).
    \end{eqnarray*}
    This is because offering the product at cost $C(x)$ will lead to
    a purchase with probability $x$, and in that case, we get $C(x)$
    revenue immediately and $L_{n-1}$ expected revenue in the future.
    Since $L_n$ is obviously increasing in $n$,
    this can be simplified further:
    \begin{eqnarray*}
        && L_n \le \max_x x(C(x) + L_n)\\
        &\implies&
        L \le 2L_n \le \max_{0<x<1} \frac{2x \cdot C(x)}{1-x}
    \end{eqnarray*}
    For $x \ge 1-\epsilon$, we have $\frac{2x \cdot
    C(x)}{1-x} \le \frac{2x \cdot m(1-x)}{1-x} \le 2m$, and for
  $x < 1 - \epsilon$, we have $\frac{2x \cdot C(x)}{1-x} \le
    \frac{2}{\epsilon}$. Either way,
    $L \le 2\max(\frac{1}{\epsilon}, m)$.

    It remains to choose $f, c$ and $q$ as per the first complexity
    condition. We use $f = 0$, $q = \min(\epsilon, \frac{1}{2m})$
    and $c = C(q) \ge \frac{1}{2}$. Indeed, if a node has more than
    0 active neighbors, it will accept a recommendation at cost
    $C(q)$ with probability $q$.

    Thus \ICM{c} has complexity at most $\frac{L}{(1-f)cq} \le
    8\max(\frac{1}{\epsilon}, m)^2$, as required.
\end{proof}

In the traditional Linear Threshold model, there are fixed
influences $b_{v,w}$ on each directed edge $(v,w)$ in the
network. Each node independently chooses a threshold $\theta$
uniformly at random from $[0,1]$, and then purchases the product
if and when the total influence on it from nodes that have
recommended the product exceeds $\theta$.

To generalize this to multiple prices, it is natural to make
$b_{v,w}$ a function $[0,1]\rightarrow[0,1]$ where $b_{v,w}(x)$
indicates the influence $v$ exerts on $w$ as a result of
recommending the product at price $x$. To simplify the exposition,
we will focus on the case where a node is equally influenced by all
its neighbors. (This is not strictly necessary but removing this
assumptions requires rephrasing the definition of $f$ to be a
{\em weighted} fraction of a node's neighbors.) Finally, we assume
for all $v,w$ that $b_{v,w}(0) = 1$ to satisfy the second general
condition for models.

\begin{defn}
    Fix a max influence function $B: (0,1]\rightarrow[0,1]$, not
    uniformly 0. We define the {\em Linear Threshold Model} \LTM{B}
    as follows:\\
    Every node independently chose a threshold $\theta$ uniformly at
    random from $[0,1]$. A node will buy the product at price $x >
    0$ only if
    $B(x) \ge \frac{\alpha}{\theta}$ where $\alpha$ denotes the
    fraction of the node's neighbors that have recommended the
    product. A node will always accept a recommendation if the
    product is offered for free.
\end{defn}

\begin{lemma}
    Fix a max influence function $B$ and let
    $K = \max_x x \cdot B(x)$. Then \LTM{B} has complexity $O(\frac{1}{K})$.
\end{lemma}

We omit the proof since it is similar to that of Lemma \ref{ICMLem}.
In fact, it is simpler since, on a line graph, a node either gets
the product for free or it has probability at most $\frac{1}{2}$ of
buying the product and passing on a recommendation.

\section{Approximating the Optimal Revenue}
\label{sec:approxalg}
In this section, we present our main theoretical contribution: a non-adaptive
seller strategy that achieves expected revenue within a constant factor of
the revenue from the optimal {\em adaptive} strategy.  We show the problem of
finding the exact optimal strategy is NP-hard (see section~\ref{sec:nphard}
in the appendix), so this kind of result is the best we can hope for.  Note
that our approximation guarantee is against the strongest possible optimum,
which is perhaps surprising: it is unclear a priori whether such a strategy
should even exist.

The strategy we propose is based on computing a {\em maximum-leaf
  spanning tree} (\maxleaf) of the underlying social network graph,
i.e., computing a spanning tree of the graph with the maximum number
of leaf nodes. The \maxleaf{} problem is known to be NP-Hard, and it
is in fact also MAX SNP-Complete, but there are several
constant-factor approximation algorithms known for the
problem~\cite{GJ79,KW91,LR98,Solis98}. In particular, one of these
is nearly linear-time~\cite{LR98}, making it practical to apply on
large online social network graphs. The seller strategy we attain
through this is an {\em influence-and-exploit} strategy that offers
the product to all of the interior nodes of the spanning tree for
free, and charges a fixed price from the leaves. Note that this
strategy works for all the buyer decision models discussed above,
including multi-price generalizations of both Independent Cascade
and Linear Threshold.

We consider the setting of Problem~\ref{problem:strategy}, where we
are given an undirected social network graph $G(V,E)$, a seed set
$S^0\subseteq V$ and a buyer decision model \M.  Throughout this
section, we will let $L$, $f$, $c$ and $q$ denote the quantities
that parametrize the model complexity, as described in Section
\ref{sec:BuyerDecision}. To simplify the exposition, we will assume
for now that the seed set is a singleton node  (i.e., $|S^0|=1$). If
this is not the case, the seed nodes can be merged into a single node,
and we can make much the same argument in that case. We will ignore
cashbacks for now, and return to address them at the end of the section.

The exact algorithm we will use is stated below:
\begin{itemize}
    \item Use the \maxleaf{} algorithm~\cite{LR98} to compute an
    approximate max-leaf spanning tree $T$ for $G$ that is rooted at
    $S_0$.
    \item Offer the product to each internal node of $T$ for free.
    \item For each leaf of $T$ (excluding $S_0$), independently flip a biased coin.
    With probability $\frac{1+f}{2}$, offer the product to the node
    for free. With probability $\frac{1-f}{2}$, offer the product to
    the node at cost $c$.
\end{itemize}
We henceforth refer to this strategy as \strategyml.

Our analysis will revolve around what we term as ``good'' vertices, defined
formally as follows:
\begin{defn}
  Given a graph $G(V,E)$, we define the {\em good} vertices to be
  the vertices with degree at least 3 and their neighbors.
\end{defn}
On the one hand, we show that if $G$ has $g$ good vertices, then the
\maxleaf{} algorithm will find a spanning tree with $\Omega(g)$
leaves. We then show that each leaf of this tree leads to
$\Omega(1)$ revenue, implying \strategyml{} gives $\Omega(g)$ revenue
overall. Conversely, we can decompose $G$ into at most $g$
line-graphs joining high-degree vertices, and the total revenue from
these is bounded by $gL = O(g)$ for all policies, which gives the
constant-factor approximation we need.

We begin by bounding the number of leaves in a max-leaf spanning
tree. For dense graphs, we can rely on the following fact
\cite{KW91,LR98}:
\begin{fact} \label{fact:maxleaf}
  The max-leaf spanning tree of a graph with minimum degree at least
  3 has at least $n/4+2$ leaves~\cite{KW91,LR98}.
\end{fact}
In general graphs, we cannot apply this result directly. However,
we can make any graph have minimum degree 3 by replacing degree-1
vertices with small, complete graphs and by contracting along edges
to remove degree-2 vertices. We can then apply Fact \ref{fact:maxleaf}
to analyze this auxiliary graph, which leads to the following result:
\begin{lemma} \label{lemma:maxleaf}
    Suppose a connected graph $G$ has $n_3$ vertices with degree at
    least $3$. Then $G$ has a spanning tree with at least
    $\frac{n_3}{8} + 1$ leaves.
\end{lemma}
\begin{proof}
    Let $n_1$ and $n_2$ denote the number of vertices of degree 1
    and 2 respectively, and let $M$ denote the number of leaves in a
    max-leaf spanning tree of $G$. If $n_1 = n_2 = 0$, the result
    follows from Fact \ref{fact:maxleaf}.

    Now, suppose $n_2 = 0$ but $n_1 > 0$. Clearly, every spanning
    tree has at least $n_1$ leaves, so the result is obvious if
    $n_1 \ge \frac{n_3}{8} + 1$. Otherwise, we replace each degree-1 vertex
    with a copy of $K_4$ (the complete graph on 4 vertices), one of whose
    vertices connects back to the rest of the graph. Let $G'$ denote the
    resulting graph. Then $G'$ has $4n_1 + n_3$ vertices, and they are all at
    least degree 3, so $G'$ has a spanning tree $T'$ with at
    least $n_1 + \frac{n_3}{4} + 2$ leaves.

    We can transform this into a spanning tree $T$ on $G$ by
    contracting each copy of $K_4$ down to a single point. Each
    contraction could transform up to 3 leaves into a single leaf,
    but it will not affect other leaves. Since there are exactly
    $n_1$ contractions that need to be done altogether, $T$ has at
    least $n_1 + \frac{n_3}{4} + 2 - 2n_1 \ge \frac{n_3}{8} + 1$
    leaves, as required.

    We now prove the result holds in general by induction on $n_2$.
    We have already shown the base case $(n_2 = 0)$. For the inductive step, we
    will define an auxiliary graph $G'$ with $n_2', n_3'$ and $M'$
    defined as for $G$. We will then show
    $n_2' < n_2, n_3' \ge n_3$, and for every spanning tree $T'$ on $G'$,
    there is a spanning tree $T$ on $G$ with at least as many leaves. This
    implies $M' \le M$, and using the inductive hypothesis, it
    follows that $M \ge M' \ge \frac{n_3'}{8} + 1 \ge \frac{n_3}{8}
    + 1$, which will complete the proof.

    Towards that end, suppose $v$ is a degree-2 vertex in $G$, and
    let its neighbors be $u$ and $w$. If $u$ and $w$ are not
    adjacent, we let $G'$ be the graph attained by contracting along
    the edge $(u,v)$. Then $n_2' = n_2 - 1$ and $n_3' = n_3$. Any
    spanning tree $T'$ on $G'$ can be extended back to a spanning
    tree $T$ on $G$ by uncontracting the edge $(u,v)$ and adding it
    to $T$. This does not decrease the number of leaves in the tree,
    so we are done.

    Next, suppose instead that $u$ and $w$ are adjacent. We cannot
    contract $(u,v)$ here since it will create a duplicate edge in $G'$.
    However, a different construction can be used. If the entire graph is
    just these 3 vertices, the lemma is trivial. Otherwise, let $G'$
    be the graph attained by adding a degree-1 vertex $x$ adjacent
    to $v$. Then $n_2' = n_2 - 1$ and $n_3' = n_3 + 1$. Now consider
    a spanning tree $T'$ of $G'$. We can transform this into a
    spanning tree $T$ on $G$ by removing the edge $(v,x)$ that must
    be in $T'$. This removes the leaf $x$ but if $v$ has degree 2
    in $T'$, it makes $v$ a leaf. In this case, $T$ and $T'$ have the same
    number of leaves, so we are done.

    Otherwise, $(u,v)$ and $(v,w)$ are also in $T'$, and since $G$
    was assumed to have more than 3 vertices, $u$ and $w$ cannot
    both be leaves in $T'$. Assume without loss of generality that
    $u$ is not a leaf. We then further modify $T$ by replacing $(v,w)$ with
    $(u,w)$. Now, $v$ is a leaf in $T$ and the only vertex whose
    degree has changed is $u$, which is not a leaf in either $T$ or
    $T'$. Therefore, $T$ and $T'$ again have the same number of
    leaves, and we are once again done.

    The result now follows from induction, as discussed above.
\end{proof}

We must further extend this to be in terms of the number of good
vertices $g$, rather than being in terms of $n_3$:
\begin{lemma} \label{lem:leaves}
  Given an undirected graph $G$ with $g$ good vertices, the
  \maxleaf{} algorithm~\cite{LR98} will construct a spanning tree
  with $\max(\frac{g}{50} + 0.5, 2)$ leaves.
\end{lemma}

\begin{proof}
  If $g = 0$, the result is trivial. Otherwise, let $n_3$ denote the
  number of vertices in $G$ with degree at least 3, and let $M$
  denote the number of leaves in a max-leaf spanning tree of $G$. By
  Lemma \ref{lemma:maxleaf}, we know $M \ge \frac{n_3}{8} + 1$.

  Now consider constructing a spanning tree as follows:
  \begin{itemize}
    \item[1.] Let $A$ denote the set of vertices in $G$ with degree at
    least 3.
    \item[2.] Set $T$ to be a minimal subtree of $G$ that connects
    all vertices in $A$.
    \item[3.] Add all remaining vertices in $G$ to $T$ one at a time. If
    a vertex $v$ could be connected to $T$ in multiple ways, connect
    it to a vertex in $A$ if possible.
  \end{itemize}
  To analyze this, note that $G - A$ can be decomposed into a
  collection of ``primitive'' paths. Given a primitive path $P$, let
  $g_P$ denote the number of good vertices on $P$ and let $l_P$
  denote the number of leaves $T$ has on $P$.

  In Step 2 of the algorithm above, exactly $n_3 - 1$ of these paths
  are added to $T$. For each such path $P$, we have $g_P \le 2$ and
  $l_P = 0$. On the remaining paths, we have $g_P = l_P$. Therefore,
  the total number of leaves on $T$ is at least
  \begin{eqnarray*}
  \sum_P l_P = (g - n_3) + \sum_P (l_P - g_P) &\ge& (g - n_3) - 2(n_3 - 1).
  \end{eqnarray*}
  Thus,
  \begin{eqnarray*}
    M &\ge& \max\left(\frac{n_3}{8} + 1, g - 3n_3 + 1\right)\\
      &\ge& \frac{24}{25} \cdot \left(\frac{n_3}{8} + 1\right) + \frac{1}{25} \cdot (g - 3n_3 + 1) = \frac{g}{25} + 1
  \end{eqnarray*}
  The result now follows from the
  fact that the \maxleaf{} algorithm gives a 2-approximation for the
  max-leaf spanning tree, and that every non-degenerate tree has at least
  two leaves.
\end{proof}

We can now use this to prove a guarantee on the performance of
\strategyml{} in terms of the number of good vertices on an arbitrary graph:
\begin{lemma} \label{lem:strategybound}
Given a social network $G$ with $g$ good vertices, \strategyml{}
guarantees an expected revenue of $\Omega((1-f)cq \cdot g)$.
\end{lemma}

\begin{proof}
  Let $T$ denote the spanning tree found by the \maxleaf{}
  algorithm. Let $U$ denote the set of interior nodes of $T$, and
  let $V$ denote the leaves of $T$ (excluding $S_0$). Since we
  assumed $|S_0| = 1$, Lemma \ref{lem:leaves} guarantees $|V| \ge
  \max(\frac{g}{50} - 0.5, 1) = \Omega(g)$.

  Note every vertex can be reached from $S_0$ by passing through
  nodes in $U$, each of which is offered the product for free. These
  nodes are guaranteed to accept the product, and therefore, they
  will collectively pass on at least one recommendation to each vertex.

  Now consider the expected revenue from a vertex $v \in V$. Let $M$
  be the random variable giving the fraction of $v$'s neighbors in
  $V$ that were {\em not} offered the product for free. We know
  $E[M] = \frac{1-f}{2}$, so with probability $\frac{1}{2}$, we have
  $M \le 1-f$.

  In this case, $v$ is guaranteed to receive recommendations from a
  fraction $f$ of its neighbors in $V$, as well as all of its
  neighbors in $U \cup S_0$ (of which there is at least 1). If we
  charge $v$ a total of $c$ for the product, it will then purchase
  the product with probability at least $q$, by the original definitions
  of $f$, $c$ and $q$. Furthermore,
  independent of $v$'s neighbors, we will ask this price from $v$
  with probability $\frac{1-f}{2}$. Therefore, our expected revenue
  from $v$ is at least $\frac{1}{2} \cdot q \cdot \frac{1-f}{2}
  \cdot c$.

  The result now follows from linearity of expectation.
\end{proof}

Now that we have computed the expected revenue from \strategyml, we need to
characterize the optimal revenue to bound the approximation ratio. This bound
is given by the following lemma.
\begin{lemma} \label{lem:optbound}
The maximum expected revenue achievable by any strategy (adaptive or
not) on a social network $G$ with $g$ good vertices is $O(L
\cdot g)$.
\end{lemma}
\begin{proof}
  Let $A$ denote the set of vertices in $G$ with degree at least 3,
  and let $n_3 = |A|$. Clearly, no strategy can achieve more than
  $n_3$ revenue directly from the nodes in $A$.

  As observed in the proof of Lemma \ref{lem:leaves}, however,
  $G - A$ can be decomposed into a collection of primitive paths. Since
  each primitive path contains at least one unique good vertex with
  degree less than 3, there is at most $g - n_3$ such paths. Even if
  each endpoint of a path is guaranteed to recommend the product,
  the total revenue from the path is at most $L$.

  Therefore, the total revenue from any strategy on such a graph is at
  most $n_3 + (g - n_3)L = O(L \cdot g)$.
\end{proof}

Now, we can combine the above lemmas to state the main theorem of
the paper, which states that \strategyml{} provides a constant factor
approximation guarantee for the revenue.
\begin{theorem} \label{thm:approx}
  Let $K$ denote the complexity of our buyer decision model \M.
  Then, the expected revenue generated by \strategyml{} on an
  arbitrary social network is $O(K)$-competitive with the expected
  revenue generated by the optimal (adaptive or not) strategy.
\end{theorem}
\begin{proof}
  This follows immediately from Lemmas \ref{lem:strategybound} and
  \ref{lem:optbound}, as well as the fact that $K =
  \frac{L}{(1-f)cq}$.
\end{proof}
As a corollary, we get the fact that the adaptivity gap is also constant:
\begin{corollary} \label{adaptGap}
  Let $K$ denote the complexity of our buyer decision model \M. Then
  the adaptivity gap is $O(K)$.
\end{corollary}

Now we briefly address the issue of cashbacks that were ignored in this
exposition. We set the cashback $r$ to be a small fraction of our expected revenue
from each individual $r_0$, i.e. $r=z\cdot r_0$, where $z < 1$. Then, our total profit
will be $n \cdot r_0 \cdot (1 - z)$. Adding this cashback decreases our total profit by a constant
factor that depends on $z$, but otherwise the argument now carries through as
before, and nodes now have a positive incentive to pass on recommendations.

In light of Corollary \ref{adaptGap}, one might ask whether the adaptivity gap
is not just 1. In other words, is there any benefit at all to be gained from using non-adaptive
strategies? In fact, there is. For example, consider a social network consisting of 4
nodes $\{v_1, v_2, v_3, v_4\}$ in a cycle, with $v_3$ connected to two other
isolated vertices. Suppose furthermore that a node will accept a recommendation with
probability 0.5 unless the price is 0, in which case the node will accept it
with probability 1. On this network, with seed set $S^0 = \{v_1\}$,
the optimal adaptive strategy is to always demand full price unless exactly one
of $v_2$ and $v_4$ purchases the product initially, in which case $v_3$ should be
offered the product for free. This beats the optimal non-adaptive strategy by a 
factor of 1.0625.

\section{Local Search} \label{sec:localsearch}
In this section, we discuss how an arbitrary seller strategy can be tweaked
by the use of a local search algorithm. Taken on its own, this technique can
sometimes be problematic since it can take a long time to converge to a good
strategy. However, it performs very well when applied to an already good
strategy, such as \strategyml. This approach of combining theoretically sound
results with local search to generate strong techniques in practice is similar
in spirit to the recent \texttt{k-means++} algorithm \cite{AV07}.

Intuitively, the local search strategy for pricing on social networks works as follows:
\begin{itemize}
    \item Choose an arbitrary seller strategy $S$ and an arbitrary node $v$ to edit.
    \item Choose a set of prices $\{p_1, p_2, \ldots, p_k\}$ to consider.
    \item For each price $p_i$, empirically estimate the expected revenue $r_i$ that is
       achieved by using the price $p_i$ for node $v$.
    \item If any revenue $r_i$ beats the current expected revenue (also estimated
       empirically) by some threshold $\epsilon$, then change $S$ to use the price $p_i$
       for node $v$.
    \item Repeat the preceding steps for different nodes until there are no more improvements.
\end{itemize}
Henceforth, we call this the \LS{} algorithm for improving seller strategies.

To empirically estimate the revenue from a seller strategy, we can always
just simulate the entire process. We know who has the product initially, we
know what price each node will be offered, and we know the probability each
node will purchase the product at that price after any number of
recommendations. Simulating this process a number of times and taking the
average revenue, we can arrive at a fair approximation at how good a strategy
is in practice. In fact, we can prove that performing local search on any
input policy will ensure that the seller gets at least as much revenue as the
original policy with high probability. The proof of this fact holds for any
simulatable input policy, and proceeds by induction on the evolution tree of
the process. The proof is somewhat technical, so we will skip it, and instead
focus on the empirical question of the advantage provided by local search.

In light of the fact that local search can only improve the revenue (and never
hurt it), it seems that one should always implement local search for any
policy.  There is a important technical detail that complicates this,
however. Suppose we wish to evaluate strategies $S_1$ and $S_2$, differing
only on one node $v$. If we independently run simulations for each strategy,
it could take thousands of trials (or more!) before the systematic change to
one node becomes visible over the noise resulting from random choices made by
the other nodes. It is impractical to perform these many simulations on a
large network every time we want to change the strategy for a single node.

Fortunately, it is possible to circumvent this problem using an observation
first noted in~\cite{KKT03}. Let us consider the
Linear Threshold model \LTM{B}. In this case, all randomness occurs before
the process begins when each node chooses a threshold that encodes how
resistant it is to buying the product. Once these thresholds have been fixed,
the entire sales process is deterministic. We can now change the strategy
slightly and {\em maintain the same thresholds} to isolate exactly what
effect this strategy change had. Any model, including Independent Cascade,
can be rephrased in terms of thresholds, making this technique possible.

The \LS{} algorithm relies heavily on this observation. While comparing
strategies, we choose several threshold lists, and simulate each strategy
against the same threshold lists.  If these lists are not representative, we
might still make a mistake drawing conclusions from this, but we will not
lose a universally good signal or a universally bad signal under the weight
of random noise.

With this implementation, empirical tests (see the next section) show the
\LS{} algorithm does do its job: given enough time, it will improve virtually
any strategy enough to be competitive.  It is not a perfect solution,
however. First of all, it can still make small mistakes while doing the
random estimates, possibly causing a strategy to become worse over
time\footnote{Note that if we choose $\epsilon$ and the number of trials
  carefully, we can make this possibility vanishingly small (this is also the
  intuition behind the local search guarantee, as we had mentioned
  earlier. In practice, however, it is usually better to run fewer trials and
  accept the possibility of regressing slightly.}.  Secondly, it is possible
to end up with a sub-optimal strategy that simply cannot be improved by any
local changes. Finally, the \LS{} algorithm can often take many steps to
improve a bad strategy, making it occasionally too slow to be useful in
practice.

Nonetheless, these drawbacks really only becomes a serious problem if one
begins with a bad strategy.  If one begins with a relatively good strategy --
for example \strategyml{} -- the \LS{} algorithm performs well, and is almost
always worth doing in practice. We justify this claim in the next section.

\subsection{Experimental Results}
In this section, we provide experimental evidence for the efficacy of the
\LS{} algorithm in improving the revenue guarantee. Note that in these
experiments, we need to assume a benchmark strategy as finding the optimal
strategy is NP-hard (see section~\ref{sec:nphard}). We pick a very simple
strategy \rand, which picks a random price independently for each node. The
results demonstrate that even this naive strategy can be coupled with the
\LS{} algorithm to do well in practice.

\begin{figure*}[htpb]
  \centering
  \subfigure[Random preferential attachment graph]{
    \includegraphics[scale=0.5,viewport=50 70 750 540,clip]{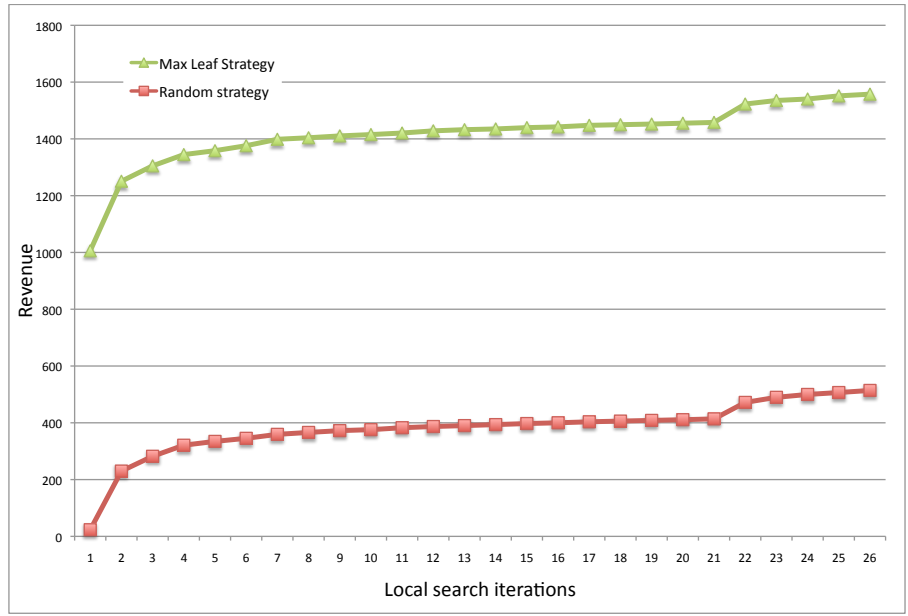}
    \label{fig:prefattachment}
  }
  \subfigure[Youtube subgraph]{
    \includegraphics[scale=0.5,viewport=50 70 750 540,clip]{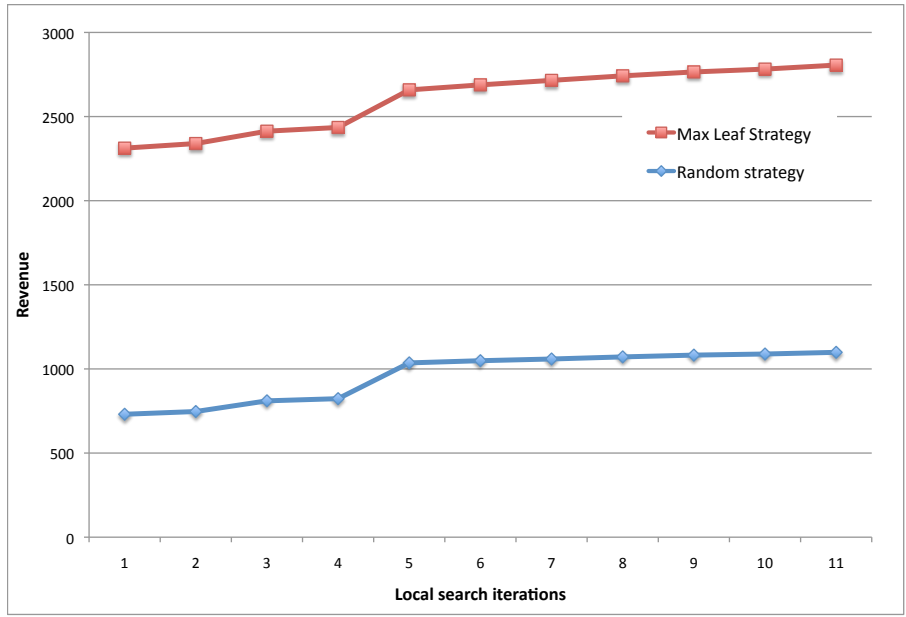}
    \label{fig:youtube}
  }
  \label{fig:revenue}
  \caption{The variation in revenue generated by \rand{} and \strategyml{} with
    the iterations of the \LS{} algorithm. The data is averaged over $10$ runs
    of a $1000$ node random preferential attachment
    graph~\subref{fig:prefattachment} or a $10000$ node subgraph of
    YouTube~\subref{fig:youtube}, starting with a random seed each time.}
\end{figure*}

We simulate the cascading process on two kind of graphs. The first graph we
study is a randomly generated graph, based on the preferential attachment
model that is a popular model for representing social
networks~\cite{NWS02}. We generate a $1000$ node preferential attachment
graph at random, and simulate the cascading process by picking a random node
as the seed in the network. The probability model we examine is a step
function (see the second example given in Lemma~\ref{ICMLem}) of
probabilities. We note that the function is necessarily arbitrary. The
result of one particular parameter settings are shown in
figure~\ref{fig:prefattachment}, which plots average revenue obtained by the
two pricing strategies: \rand{} and \strategyml. Each point on the figure is
obtained by average revenue over 10 runs on the same graph but with a
different (random) seed.  The horizontal axis indicates the number of \LS{}
iterations that were done on the graph, where each iteration consisted of
simulating the process 50 times, and choosing the best value over the
runs. It is clear from the graph that \strategyml{} does quite well even
without the addition of \LS{}, although the addition of \LS{} does increase
the revenue. On the other hand, the \rand{} strategy performs poorly on its
own, but its revenue increases steadily with the iterations of the \LS{}
algorithm. We note that the difference between the revenue from the two
policies does vary (as expected) with the probability model, and the
difference between the revenue is not as large in all the different runs. But
the difference does persist across the runs, especially when the strategies
are run without the local search improvement.

We also conduct a similar simulation with a real-world network, namely the
links between users of the video-sharing site YouTube.\footnote{The network
  can be freely downloaded; see~\cite{MMGDB07} for details.}
The YouTube network has millions of nodes, and we only study a subset of
$10,000$ nodes of the network. We simulate the random process as earlier, and
the results are shown in figure~\ref{fig:youtube}. Again, we note that
\strategyml{} does very well on its own, easily beating the revenue of
\rand. The \rand{} strategy does improve a lot with \LS, but it fails to
equalize the revenue of \strategyml. The large size of the YouTube graph and
the expensive nature of the \LS{} algorithm restrict the size of the
experiments we can conduct with the graph, but the results from the above
does experiments do offer some insights. In particular, \strategyml{}
succeeds in extracting a good portion of the revenue from the graph, if we
consider the revenue obtained from \strategyml{} combined with \LS{} based
improvements to be the benchmark.  Further, \LS{} can improve the revenue
from any strategy by a substantial margin, though it may not be able to
attain enough revenue when starting with a sub-optimal strategy such as
\rand. Finally, we observe that the combination of \strategyml{} and \LS{}
generates the best revenue among our strategies, and it is an open question
as to whether this is the optimal adaptive strategy.

\section{Conclusions}
In this work, we discussed pricing strategies for sellers distributing a
product over social networks through viral marketing.  We show that computing
the optimal (one that maximizes expected revenue) non-adaptive strategy for a
seller is NP-Hard. In a positive result, we show that there exists a
non-adaptive strategy for the seller which generates expected revenue that is
within a constant factor of the expected revenue generated by the optimal
adaptive strategy. This strategy is based on an {\em influence-and-exploit}
policy which computes a max-leaf spanning tree of the graph, and offers the
product to the interior nodes of the spanning tree for free, later on
exploiting this influence by extracting its profit from the leaf nodes of the
tree. The approximation guarantee of the strategy holds for fairly general
conditions on the probability function.

\section{Open Questions}
The added dimension of pricing to influence maximization models poses a host
of interesting questions, many of which are open. An obvious direction in
which this work could be extended is to think about influence models stronger
than the model examined here. It is also unclear whether the assumptions on
the function $C(\cdot)$ are the minimal set that is required, and it would be
interesting to remove the assumption that there exists a price at which the
probability of acceptance is 1. A different direction of research would be to
consider the game-theoretic issues involved in a practical system. Namely, in
the model presented here, we think of each buyer as just sending the
recommendations to all its friends and ignore the issue of any ``cost''
involved in doing so, thereby assuming all the nodes to be non-strategic. It
would be very interesting to model a system where the nodes were allowed to
behave strategically, trying to maximize their payoff, and characterize the
optimal seller strategy (especially w.r.t. the cashback) in such a setting.

\section{Acknowledgments}
Supported in part by NSF Grant ITR-0331640, TRUST (NSF award number
CCF-0424422), and grants from Cisco, Google, KAUST, Lightspeed, and
Microsoft. The third author is grateful to Mukund Sundararajan and Jason
Hartline for useful discussions.

\bibliographystyle{plain}
\bibliography{references}

\section{Appendix}
\vspace{2mm}
\subsection{Hardness of finding the optimal strategy}
\label{sec:nphard}
In this section, we show that Problem~\ref{problem:strategy} is
NP-hard even for a very simple buyer model \M{} by a reduction from
vertex cover with bounded degree (see~\cite{GJ79} for the hardness of
bounded-degree vertex cover). Letting $d$ denote the degree
bound, and letting $p = \frac{1}{4d}$, we will use an Independent
Cascade Model \ICM{C} with:
\begin{eqnarray*}
    C(x) = \left\{
      \begin{array}{ll}
        1 & \textrm{ if } x < p,\\
        0 & \textrm{ if } x \ge p
      \end{array}\right.
\end{eqnarray*}

Intuitively, the seller has to partition the nodes into ``free''
nodes and ``full-price'' nodes. In the former case, nodes are offered
the product for free, and they accept it with probability 1 as soon
as they receive a recommendation. In the latter case, nodes are
offered the product for price 1, and they accept each recommendation
with probability $p$. (Note that the seller is {\em allowed} to use
other prices between $0$ and $1$ but a price of $1$ is always
better.)

We are going to use a special family of graphs illustrated in
Figure~\ref{fig:reduction}. The graph consists of four layers:
\begin{itemize}
\item A singleton node $s$, which we will use as the only initially
  active node (i.e., $S^0=\{s\}$);
\item $s$ links to a set of $n$ nodes, denoted by $V_1$;
\item Nodes in $V_1$ also link to another set of nodes, denoted by $V_2$.
Each node in $V_1$ will be adjacent to $d$ nodes in $V_2$, and each
node in $V_2$ will be adjacent to $2$ nodes in $V_1$ (so
$|V_2|=dn/2$);
\item Each node $v\in V_2$ also links to $k = 20d$ new nodes, denoted by
$W_v$; these $k$ nodes do not link to any other nodes. The union of
all $W_v$'s is denoted by $V_3$.
\end{itemize}

\begin{figure}[htpb]
  \centering
  \includegraphics[scale=0.7]{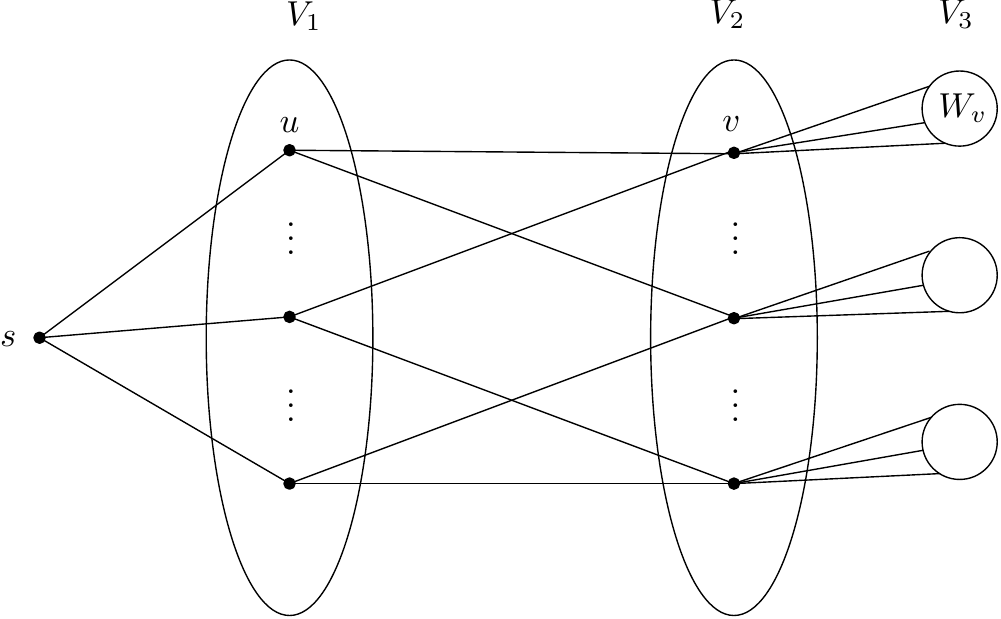}
  \caption{Reducing Bounded-Degree Vertex Cover to Optimal Network Pricing}
  \label{fig:reduction}
\end{figure}

We first sketch the idea of the hardness proof. The connection
between $V_1$ and $V_2$ will be decided by the vertex cover
instance: given a vertex cover instance $G'(V,E)$ with bounded
degree $d$, we construct a graph $G$ as above where $V_1=V$ and
$V_2=E$, adding an edge between $V_1$ and $V_2$ if the corresponding
vertex is incident to the corresponding edge in $G'$. The key lemma
is that, in the optimal pricing strategy for $G$, the subset of nodes
in $V_1$ that are given the product for free is the minimum set that
covers $V_2$ (i.e., a minimum vertex cover of $G'$).

To formalize this, first note that, in an optimal strategy, all nodes
in $V_3$ should be full-price. Giving the product to them for free
gets 0 immediate revenue, and offers no long-term benefit since nodes
in $V_3$ cannot recommend the product to anyone else. If the nodes
are full-price, on the other hand, there is at least a chance at
some revenue.

On the other hand, we show the optimal strategy must also ensure each
vertex in $V_2$ eventually becomes active with probability 1.
\begin{lemma} \label{lem:v2}
  In an optimal strategy, every node $v \in V_2$ is free, and can be
  reached from $s$ by passing through free nodes.
\end{lemma}
\begin{proof}
Suppose, by way of contradiction, that the optimal strategy has a node
$v \in V_2$ that does not satisfy these conditions. Let $u_1$ and
$u_2$ be the two neighbors of $v$ in $V_1$, and let $q$ denote the
probability that $v$ eventually becomes active.

We first claim that $q < 2dp$. Indeed, if $v$ is full-price, then
even if $u_1$ and $u_2$ become active, the probability that $v$
becomes active is $1 - (1-p)^2 < 2p$. Otherwise, $u_1$ and $u_2$ are
both full-price. Since $u_1$ and $u_2$ connect to at most $2d$ edges
other than $v$, the probability that one of them becomes active
before $v$ is at most $1 - (1-p)^{2d} < 2dp$. Thus, $q < 2dp$.

It follows that the total revenue that this strategy can achieve from
$u_1$, $v$, and $W_v$ is $2 + kqp < 2 + 2kdp^2 = 4.5$. Conversely,
if we make $u_1$ and $v$ free, we can achieve $kp = 5$ revenue from
the same buyers. Furthermore, doing this cannot possibly lose
revenue elsewhere, which contradicts the assumption that our
original strategy was optimal.
\end{proof}

It follows that, in an optimal strategy, all of $V_3$ is full-price,
all of $V_2$ is free, and every node in $V_2$ is adjacent to a free
node in $V_1$. It remains only to determine $C$, the nodes in $V_1$,
that an optimal strategy should make free. At this point, it should be
intuitively clear that $C$ should correspond to a minimum
vertex-cover of $V_2$. We now formalize this as follows:

\begin{lemma} \label{lem:v1}
    Let $C$ denote the set of free nodes in $V_1$, as chosen by an
    optimal strategy. Then $C$ corresponds to a minimum vertex cover
    of $G'$.
\end{lemma}
\begin{proof}
    As noted above, every node in $V_2$ must be adjacent to a node
    in $C$, which implies $C$ does indeed correspond to a vertex
    cover in $G'$.

    Now we know an optimal strategy makes every node in $V_2$ free,
    and every node in $V_3$ full-price. Once we know $C$, the strategy
    is determined completely. Let $x_C$ denote the expected revenue
    obtained by this strategy. Since all nodes in $V_2$ are free and
    are activated with probability $1$, we know the strategy achieves
    0 revenue from $V_2$ and $p|V_3|$ expected revenue from $V_3$.

    Among nodes in $V_1$, the strategy achieves 0 revenue for free
    nodes, and exactly $1 - (1-p)^{d+1}$ expected revenue for each
    full-price node. This is because each full-price node is
    adjacent to exactly $d+1$ other nodes, and each of these nodes
    is activated with probability 1. Therefore,
    $x_C = (|V_1| - |C|)(1 - (1-p)^{d+1}) + p|V_3|$, which is clearly
    minimized when $C$ is a minimum-vertex cover.
\end{proof}

Therefore, optimal pricing, even in this limited scenario, can be
used to calculate the minimum-vertex cover of any bounded-degree
graph, from which NP-hardness follows.

\begin{theorem} \label{lem:twocoupon}
  Two Coupon Optimal Strategy Problem is NP-Hard.
\end{theorem}

\end{document}